\documentclass{IEEEtran4PSCC}

\usepackage{graphicx}
\usepackage{eurosym}
\usepackage{amssymb}
\usepackage{amsfonts}
\usepackage{epstopdf}
\usepackage{epsf,subfigure}
\usepackage{psfrag}
\usepackage{graphics}
\usepackage{color} 
\usepackage{cite}
\usepackage{array,multirow,pbox}
\usepackage{caption}

\newcommand{\rem}[1]{}
\newtheorem{proposition}{Proposition}
\newtheorem{definition}{Definition}
\newtheorem{theorem}{Theorem}
\newtheorem{example}{Example}

\newtheorem{remark}{Remark}
\newtheorem{assumption}{Assumption}

\newenvironment{proof}[1][Proof]{\begin{trivlist}
\item[\hskip \labelsep {\bfseries #1}]}{\end{trivlist}}

\newcommand{\qed}{\nobreak \ifvmode \relax \else
      \ifdim\lastskip<1.5em \hskip-\lastskip
      \hskip1.5em plus0em minus0.5em \fi \nobreak
      \vrule height0.75em width0.5em depth0.25em\fi}

\newcommand{\mysubeq}[2]{
\begin{subequations}\label{#1}
\begin{align}
#2
\end{align}\end{subequations}}

\graphicspath{{Figures/}}

\ifCLASSINFOpdf
\else
\fi
\usepackage[cmex10]{amsmath}
\begin{document}
%
\title{Approximating Flexibility in Distributed Energy Resources: A Geometric Approach}

\author{
\IEEEauthorblockN{Soumya Kundu, Karanjit Kalsi}
\IEEEauthorblockA{Optimization and Control Group\\ 
Pacific Northwest National Laboratory\\ 
Richland, WA 99352 USA\\
Email: \{soumya.kundu, karanjit.kalsi\}@pnnl.gov}
\and
\IEEEauthorblockN{Scott Backhaus}
\IEEEauthorblockA{Information Systems and Modeling Group\\ 
Los Alamos National Laboratory\\ 
Los Alamos, NM 87545 USA\\ 
Email: backhaus@lanl.gov}
}


\maketitle

\begin{abstract}
With increasing availability of communication and control infrastructure at the distribution systems, it is expected that the distributed energy resources (DERs) will take an active part in future power systems operations. One of the main challenges associated with integration of DERs in grid planning and control is in estimating the available flexibility in a collection of (heterogeneous) DERs, each of which may have local constraints that vary over time. In this work, we present a geometric approach for approximating the flexibility of a DER in modulating its active and reactive power consumption. The proposed method is agnostic about the type and model of the DERs, thereby facilitating a plug-and-play approach, and allows scalable aggregation of the flexibility of a collection of (heterogeneous) DERs at the distributed system level. Simulation results are presented to demonstrate the performance of the proposed method.
\end{abstract}

\begin{IEEEkeywords}
Demand response, load aggregation, Minkowski sum, polynomial optimization.
\end{IEEEkeywords}

\thanksto{This work was supported by the United States Department
of Energy under the Grid
Modernization Lab Consortium initiative.}

\section{Introduction}

{T}{raditionally} the bulk of the responsibility of maintaining the (real-time) balance between load and generation has rest upon the conventional generators, e.g. spinning reserves, which provide various forms of grid ancillary services such as inertia support, regulation and ramping. As the power systems transition towards a greener grid with larger penetration of renewable generation, a large fraction of which is expected to be \textit{distributed}, there is a general consensus that various forms of distributed energy resources (DERs), including flexible and responsive electrical loads, need to be coordinated and controlled in real-time to provide grid support. This requires an appropriate understanding of the loads behavior, including their physical models, at the details useful for real-time coordination. It is important that such models capture the necessary information, such as the associated dynamics (thermal dynamics of a residential air-conditioner), measure of the available flexibility (reactive power capacity of a grid-connected inverter), the end-user constraints (desired state-of-charge of an electric vehicle battery), etc., while being tractable for real-time operation.

Modeling of an ensemble of flexible loads for ancillary services (in particular, frequency regulation and ramping) have been explored in the literature in recent years \cite{Callaway:2009, Kundu:2012CDC, Perfumo:2012, Mathieu:2013, Zhang:2013, Hao:2015, Sanandaji:16, Zhao:2016}. The methods proposed in these articles are applicable to ensembles of similar loads (albeit with heterogeneous parameters), such as a collection of residential air-conditioners, or a collection of plug-in electric vehicles. Such types of aggregate flexibility models are suitable for a transmission system operator which views the net demand flexibility available at the distribution-level as a lumped-model. Aggregation at this level (tens of thousands of loads) do not explicitly take into account the operational constraints (line-flow and voltage limits) at the mid/low-voltage distribution systems. As the fraction of flexible loads increase, however, chances of violation of operational constraints due to control of flexible demand will increase. 

In order to efficiently coordinate tens of thousands of flexible loads in distribution systems, while also satisfying line-flow and node voltage constraints, hierarchical modeling and control frameworks \cite{Callaway:2011,Bernstein:2015} become attractive. Load aggregators, referred to as \textit{aggregate device controller} (or ADC), help in keeping the size of control problem tractable by aggregating the neighboring DERs locally. This aggregation can be done at the level of a couple of service transformers (tens of residential customers). It is the responsibility of the ADC to capture the aggregated flexibility, in terms of active and reactive power, of the local DERs which are likely to be of different types (e.g. a collection of air-conditioners, electric water-heaters, batteries, solar photovoltaic inverers wind inverters) and ratings. 

The rest of the paper is organized as follows. In Section\,\ref{S:problem} we describe the problem of aggregating flexibility of heterogeneous DERs. Section\,\ref{S:back} presents some key concepts that form the basis of our work. In Section\,\ref{S:method} we discuss in details the geometric programming approach to the aggregation of flexibility, while exploring the metrics of quality of approximation in Section\,\ref{S:metric}. Numerical results illustrating the concept are presented in Section\,\ref{S:result}, before concluding the article in Section\,\ref{S:concl}.

\section{Problem Description}\label{S:problem}

Consider a hierarchical distribution system operation framework in which a distribution system operator (DSO) communicates with aggregate device controllers (ADCs) at every control period and solves an optimal dispatch problem to schedule the DERs. It is assumed that the ADCs are so placed that the operational constraints (line-flow and voltage limits) are trivially satisfied at the ADC-level. The primary responsibilities of the ADC are two-fold: 1) before the start of each control period (5-15\,min), the ADC estimates the net flexibility (in active and reactive power) available among the DERs under its control, and exchanges that information to the DSO; 2) during the control period, the ADC coordinates the DERs in real-time to track the active and reactive power set-points dispatched by the DSO. In this paper we restrict our discussion only on the first part, i.e. the problem of estimating the net flexibility in a group of (dissimilar) DERs. 

\begin{figure*}[thpb]
\centering
\captionsetup{justification=centering}
\subfigure[batteries]{
\includegraphics[scale=0.35]{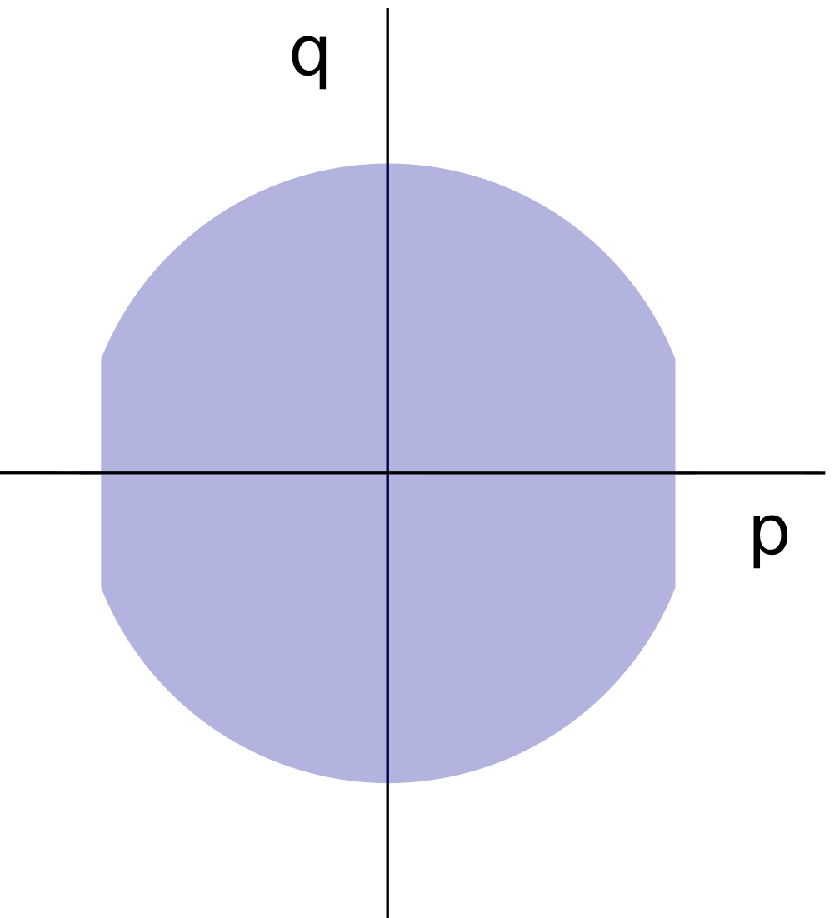}\label{F:battery}
}
\hspace{0.1in}
\subfigure[switching loads]{
\includegraphics[scale=0.35]{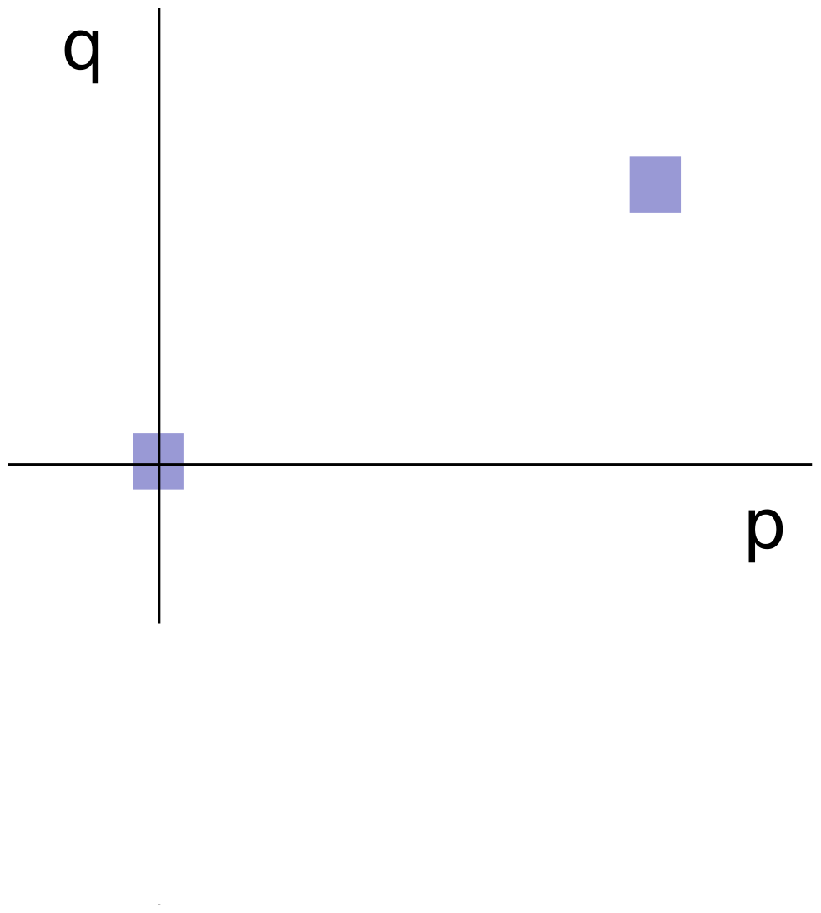}\label{F:hvac}
}
\hspace{0.1in}
\subfigure[PV inverters]{
\includegraphics[scale=0.35]{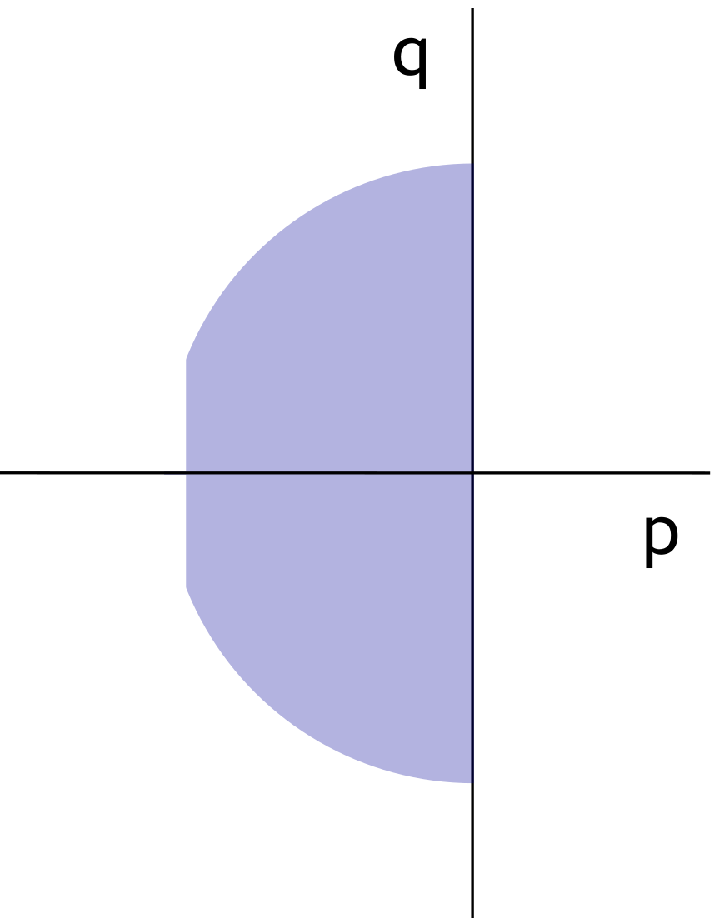}\label{F:pv}
}
\hspace{0.1in}
\subfigure[wind inverters]{
\includegraphics[scale=0.35]{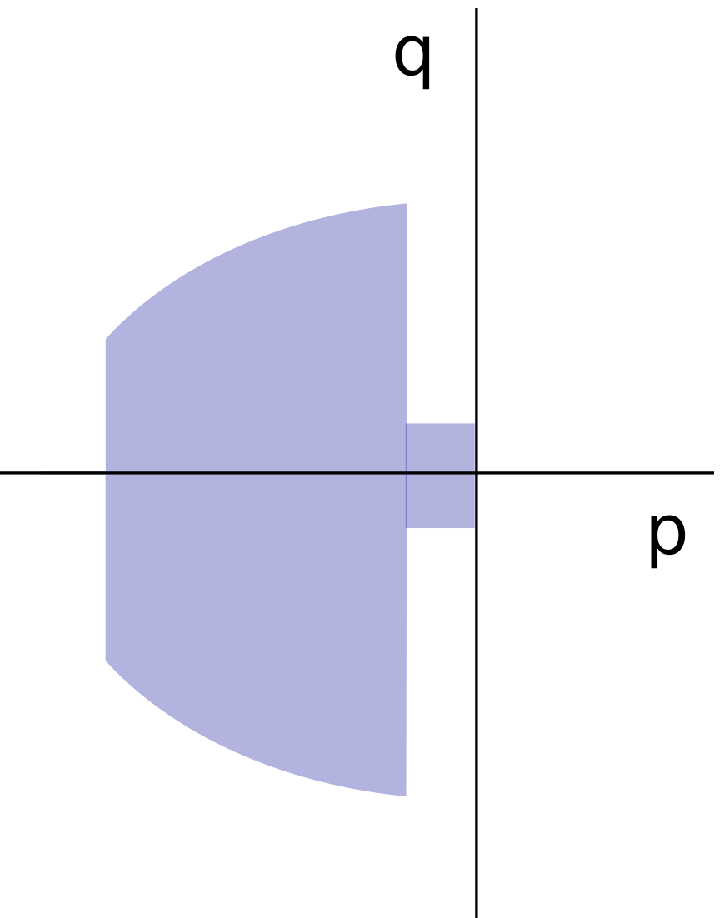}\label{F:wind}
}
\caption[Optional caption for list of figures]{Representation of continuous and discrete flexibility domains for an individual DER of certain types.}
\label{F:domains}
\end{figure*}
For the purpose of this article, any energy (consuming or generating) resource which offers certain flexibility in active ($p$) and/or reactive ($q$) power, possibly over a reasonably short time window (such as a 5-15\,min long control period), is considered as a DER. We use the notation $\mathcal{F}$ to represent the flexibility domain as a collection of $(p,q)$-points that are physically admissible by the DER (possibly via some local device-level control). Note that the flexibility domain ($\mathcal{F}$) could be a continuous or a discrete domain. For example, a solar photovoltaic (PV) inverter that can modulate its active and reactive power over a continuous range will have a continuous flexibility domain, while the flexibility domain for switching loads, such as an air-conditioner or electric water-heater, will be discrete. Moreover this flexibility is time-varying, and depends on exogenous parameters as well as end-user preferences. 

Fig.\,\ref{F:domains} shows examples of flexibility domains for certain types of DERs, with positive (negative) values of $p$ and $q$ denoting consumption (generation). Discrete flexibility domain of a switching load (e.g. air-conditioner) that operates in two discrete operational states (`on' and `off') is shown in Fig.\,\ref{F:hvac}, while the rest of the plots represent continuous flexibility domains. Batteries (Fig.\,\ref{F:battery} offer full four-quadrant flexibility, while PV (Fig.\,\ref{F:pv}) and wind (Fig.\,\ref{F:wind}) inverters offer flexibility only on the left half-plane (active power generation). It must be noted that the flexibility offered by the DERs is dynamic, and change based on end-usage and exogenous influence. For example, if there is a cloudy sky, the PV inverter output might only be restricted to a small fraction of its rated generation. Similarly, the air-conditioner may be forced to operate mostly in `on'-state, if the outside air-temperature is high. Finally, the flexibility domains should be able to capture DER uncertainties, which could be modeled in the robust sense via a conservative estimate of the available flexibility.

In this article, we will focus our discussion on the challenges of aggregating the flexibility of a heterogeneous mix of DERs at the ADC, while the issues regarding the dynamic evolution of flexibility and the associated uncertainties will be addressed in future work. Let us assume that there are $N$ DERs with an ADC. The flexibility domain of the $i$-th DER is denoted by $\mathcal{F}_i$\,, such that its active and reactive power consumption (with negative value signifying net generation)
\begin{align*}
(p_i,q_i)\in\mathcal{F}_i\,,\quad i\in\lbrace 1,2,\dots,N\rbrace\,.
\end{align*}
The goal of the flexibility aggregation task is to find the net flexibility domain $\mathcal{F}$ in the form of a Minkowski  sum of the individual DER flexibility domains ($\mathcal{F}_i$), such that,
\begin{align*}
\mathcal{F}&:=\biguplus_{i=1}^N\mathcal{F}_i=\left\lbrace (p,q)\,\left| \begin{array}{c} p=\sum_{i=1}^Np_i\\
q=\sum_{i=1}^Nq_i\\
\,(p_i,q_i)\in\mathcal{F}_i~\forall i\end{array}\right.\right\rbrace.
\end{align*}
Calculating the exact Minkowski sum of the individual flexibility domains is computationally complex, especially as the number of DERs increases. Thus, from a practical point-of-view, a desirable approach is to construct approximations of the aggregate flexibility domain at the ADC in a scalable way. In the rest of the paper, we will discuss methods of computing the approximated aggregate flexibility at the ADC using geometric optimization procedures.

\section{Background}\label{S:back}

\subsection{Homothetic Transformation}
In a recent work, \cite{Zhao:2016}, authors used a geometric optimization-based approach to compute the \textit{inner} and \textit{outer} polytopic approximation of the aggregated flexibility of an ensemble of thermostatically-controlled loads, in the two-dimensional space of control variable and active power consumption. We propose to apply a similar approach to approximate the aggregated flexibility ($\mathcal{F}_i$) in active and reactive power consumption. Before we explain the approach, let us first introduce the concept of \textit{homothets}.

\begin{definition}
\cite{Zhao:2016} A \textit{homothet} of a compact convex domain $\mathcal{F}\subset\mathbb{R}^n$ is defined as the family of domains which can be expressed as $\mathcal{H}[\alpha,\beta;\mathcal{F}]:=\alpha\,\mathcal{F}+\beta$, where $\alpha$ is a positive scalar and $\beta\in\mathbb{R}^n$. Henceforth, the scalar $\alpha$ will be referred to as the \textit{`scaling factor'} and $\beta$ as the \textit{`translational vector'}.
\end{definition} 

Thus a \textit{homothetic} transformation is a transformation applied on the state-space ($\mathbb{R}^n$) using uniform scaling and translation. The following result is useful in computing the inner and outer approximations of a set of homothets of a given domain (extension of the following result for more than two homothets follows trivially),

\begin{theorem}\label{T:Minkowski}
\cite{Schneider:2013} Consider a compact convex domain $\mathcal{F}\subset\mathbb{R}^n$\,, and two of its homothets $\mathcal{H}[\alpha_i,\beta_i\;\mathcal{F}]$\,, $\alpha_i>0$\,, $\beta_i\in\mathbb{R}^n$ $\forall i\in\lbrace 1,2\rbrace\,$. The Minkowski sum of the homothets is given by 
$\mathcal{H}[\alpha,\beta;\mathcal{F}]=\biguplus_{i}^{2}\mathcal{H}[\alpha_i,\beta_i;\mathcal{F}]$\,, with the \textit{scaling factor} as $\alpha=\alpha_1+\alpha_2$ and the \textit{translational vector} as $\beta=\beta_1+\beta_2$\,.
\end{theorem}

\subsection{Sum-of-Squares Programming}

\begin{definition}
Any n-variate polynomial that can be expressed as sum of squared polynomials, is called a sum-of-squares (SOS) polynomial. We denote the ring of all SOS polynomials in $x\in\mathbb{R}^n$ by $\Sigma[x]\,.$
\end{definition}

Positivstellensatz theorem allows one to translate a set of semi-algebraic constraints into SOS feasibility conditions. In particular, Putinar's Positivstellensatz theorem says:
\begin{theorem}\label{T:Putinar}
\cite{Putinar:1993,Lasserre:2009} Consider a compact domain $\mathcal{K}\!=\!\!\left\lbrace x\!\in\!\mathbb{R}^n\!\left|\, g_i(x)\!\geq\! 0\,\forall i\!=\!1,\dots,m\right.\right\rbrace$, where $g_i(x)\,\forall i$ are polynomials and $g_i(\cdot)\geq 0\,\forall i$ define compact domains. Then a polynomial $f(x)$ is positive on $\mathcal{K}$ if and only if there exist SOS polynomials $\sigma_i(x)$ such that
$f(x) - \sum_i \sigma_i(x)g_i(x)$ is SOS.
\end{theorem}

For every SOS polynomial $f(x)$ of degree $2d$ ($d$ is a positive integer), there exists a positive semi-definite matrix $\Xi$ such that the Gramm matrix representation $f(x)\!=\!z(x)^T\Xi\,z(x)$ holds, where $z(x)$ is a vector of monomials in $x$ of degree less than or equal to $d$ \cite{Parrilo:2000}. Thus each SOS problem can be cast into an equivalent semidefinite programming (SDP) problem and solved via SOSTOOLS \cite{sostools13}, in conjunction with SDP solvers such as SeDuMi \cite{Sturm:1999}. Next we demonstrate the application of SOS programming to compute the approximations of flexibility domains using prototypes.

\section{Aggregate Flexibility: Geometric Approach}\label{S:method}

In this section we describe how one can use Theorem\,\ref{T:Minkowski} to design an algorithm to compute the outer and inner approximation of the aggregated flexibility without explicitly computing their Minkowski sum. The idea is to first define a \textit{`prototype'} domain $\mathcal{F}^0\subset\mathbb{R}^2$, either as a polytope (e.g. a \textit{box} constraint) or a generic convex domain, and obtain the \textit{outer} and \textit{inner} approximations of the individual flexibility domains ($\mathcal{F}_i\subset\mathbb{R}^2$) of the $i$-th ADC as the \textit{homothets} of the \textit{prototype}. For each $i\in\lbrace 1,2,\dots,N\rbrace$\,, we seek to find positive scalars $\underline{\alpha_i}\,,\,\overline{\alpha_i}$\,, and vectors $\underline{\beta_i}\,,\,\overline{\beta_i}\in\mathbb{R}^2$, such that
\begin{align}\label{E:indiv_approx}
\mathcal{H}\!\left[\underline{\alpha_i},\underline{\beta_i};\mathcal{F}^0\right]\,\subseteq\,\mathcal{F}_i \,\subseteq\, \mathcal{H}\!\left[\overline{\alpha_i},\overline{\beta_i};\mathcal{F}^0\right].
\end{align}
Computation of the outer and inner approximations would be discussed in details in the following sub-section. Note that,
\begin{proposition} \label{prop:aggr}
If the flexibility domains of a set of DERs are outer and inner approximated by \textit{homothets} of a compact convex \textit{prototype} domain $\mathcal{F}^0\!\subset\!\mathbb{R}^2$ as in \eqref{E:indiv_approx}, then their aggregate flexibility domain $\mathcal{F}\!\subset\!\mathbb{R}^2$ is outer and inner approximated as follows,
\begin{subequations}\begin{align*}
&\mathcal{H}\!\left[\underline{\alpha},\underline{\beta};\mathcal{F}^0\right]\,\subseteq\,\mathcal{F} \,\subseteq\, \mathcal{H}\!\left[\overline{\alpha},\overline{\beta};\mathcal{F}^0\right]\,\\
\text{where }&\underline{\alpha}=\sum_{i=1}^N\underline{\alpha_i}\,,\,\underline{\beta}=\sum_{i=1}^N\underline{\beta_i}\,,\,\overline{\alpha}=\sum_{i=1}^N\overline{\alpha_i}\,,\,\overline{\beta}=\sum_{i=1}^N\overline{\beta_i}\,.
\end{align*}\end{subequations}
\end{proposition} 
\begin{proof}
The proof follows trivially from Theorem\,\ref{T:Minkowski}.\hfill\qed
\end{proof}

Let us now discuss how to construct the \textit{outer} and \textit{inner} approximations of flexibility domains using \textit{prototypes}. In this paper, we consider the two types of flexibility domains: either 1) the flexibility region is a compact domain with a boundary defined by piece-wise polynomials, or 2) the flexibility domain is a collection of admissible points in the $(p,q)$-space. Note that, we allow non-convex representations of the flexibility domains. 
%
%
%
In this paper, we consider flexibility domains that can be represented in the generic form of:
\begin{align}\label{E:flex_compact}
\mathcal{F}&=\bigcup_{k=1}^K\mathcal{F}^k\\
\text{where }~\mathcal{F}^k&=\left\lbrace (p,q)\left|\, g_1^k(p,q)\!\geq\! 0,\dots,\,g_{m^k}^k(p,q)\!\geq\! 0\right.\right\rbrace\!,\!
\end{align} 
where $g_j^k(p,q)\,\forall k\in\lbrace 1,2,\dots,m^k\rbrace\,\forall i$ are polynomials. This representation can be used to represent the continuous and discrete flexibility domains  of the types depicted in Fig.\,\ref{F:domains}, as explained in the following examples. 

\begin{example}
(\textsc{Batteries}) The flexibility domain of a battery, with a maximum charge/discharge rate of $p^{\max}$ and the apparent power rating of $s\!>\!p^{\max}$\,, is given by 
\begin{align*}
\mathcal{F}\!=\!\left\lbrace (p,q)\left|p\in[-p^{\max},{p}^{\max}],\,|q|\leq\sqrt{s^2-p^2}\right.\right\rbrace,
\end{align*} 
which could be expressed in the form of \eqref{E:flex_compact}, with $g_1(p,q)\!=\!s^2\!-\!p^2\!-\!q^2,$ and $g_2(p,q)\!=\!\left(p^{\max}\right)^2\!-\!p^2$. \hfill\qed\end{example} 

\begin{example}
(\textsc{PV Inverters}) The flexibility domain of a battery unit, with a maximum active power generation of $p^{\max}$ and the apparent power rating of $s\!>\!p^{\max}$\,, is given by
\begin{align*}
\mathcal{F}\!=\!\left\lbrace (p,q)\left|\,p\!\in\![-p^{\max},0],\, |q|\!\leq\! \sqrt{s^2\!-\!p^2}\right.\right\rbrace,
\end{align*} 
which could be expressed in the form of \eqref{E:flex_compact}, with $g_1(p,q)\!=\!s^2\!-\!p^2\!-\!q^2,$ and $g_2(p,q)\!=\!-\!p^{\max} p\!-\!p^2$. \hfill\qed\end{example} 

\begin{example}
(\textsc{Wind Inverters}) The flexibility domain of a wind inverter, with a maximum active power generation of $p^{\max}$ and the apparent power ratings $s_1\!>\!\sqrt{\alpha}p^{\max}$ (due to rotor current limits) and $s_2\!>\!\sqrt{\alpha}p^{\max}$ (due to stator current limits)\,, for some $\alpha>0$ is given by \cite{Lund:2007,Tian:2013,Martin:2015} 
\begin{align*}
\mathcal{F}&=\mathcal{F}^1\bigcup\mathcal{F}^2\bigcup\mathcal{F}^3\\
\mathcal{F}^1&=\left\lbrace (p,q)\left| \,p\in[-p^0,0]\,,\,q\in[-q^0,\,q^0]\right.\right\rbrace\\
\mathcal{F}^2&=\left\lbrace (p,q)\left|
p\in[-{p}^{\max},-p^0),\,0\leq q\leq\sqrt{s_2^2-\alpha p^2}\right.\right\rbrace\\
\mathcal{F}^2&=\left\lbrace (p,q)\left|
p\in[-{p}^{\max},-p^0),\,-\sqrt{s_1^2-\alpha p^2}\leq q\leq 0\right.\right\rbrace
\end{align*} 
where $p^0$ and $q^0$ are much smaller than the rated capacities. The flexibility can be expressed in the form of \eqref{E:flex_compact}, with 
\begin{align*}
g_1^1(p,q)&\!=\!-p^2\!-\!p\,p^0,~g_2^1(p,q)\!=\!(q^0)^2\!-\!q^2,\,\\
g_1^2(p,q)&\!=\!-p^2\!-\!p\,\left(p^0\!+\!p^{\max}\right)\!-\!p^0p^{\max},\,\\
g_2^2(p,q)&\!=\!q,\,g_3^2(p,q)\!=\!s_2^2\!-\!\alpha p^2\!-\!q^2,\,\\
g_1^3(p,q)&\!=\!g_1^2(p,q),\,g_2^3(p,q)\!=\!-q,\,g_3^3(p,q)\!=\!s_1^2\!-\!\alpha p^2\!-\!q^2.\qed
\end{align*}\end{example} 

\begin{example}\label{Ex:hvac}
(\textsc{Air-Conditioners}) The flexibility domain of a residential air-conditioner with an active power consumption rating of $p^{\max}$ (equal to the power consumed in `on' state) is represented by 
\begin{align*}
\mathcal{F}&=\mathcal{F}^1\bigcup\mathcal{F}^2\\
\mathcal{F}^1&=\left\lbrace (0,0)\right\rbrace,\,\mathcal{F}^2=\left\lbrace ({p}^{\max},\gamma {p}^{\max})\right\rbrace
\end{align*} 
where $\gamma>0$ is related to the power factor. One possible way of representing the flexibility in the form of \eqref{E:flex_compact} is by choosing $g_1^1(p,q)\!=\!-p^2,\,g_2^1(p,q)\!=\!-q^2,\,g_1^2(p,q)\!=\!-\left(p-p^{\max}\right)^2$ and $g_2^2(p,q)\!=\!-\left(q-\gamma p^{\max}\right)^2$\,. \hfill\qed\end{example}

In the scope of this work, we will focus our attention to \textit{prototype} domains that are convex polygons, expressed as
\begin{align}\label{E:proto}
\!\!\mathcal{F}^0\!=\!\left\lbrace (p,q)\left|\, A\begin{bmatrix}
p\\q
\end{bmatrix}\!\leq\!b\,,\,A\!=\!\!\begin{bmatrix}
a_{1p}&a_{1q}\\\vdots\\a_{np}&a_{nq}
\end{bmatrix}\!,\,b\!=\!\!\begin{bmatrix}
b_1\\\vdots\\b_n
\end{bmatrix}\! \right.\right\rbrace\!.\!
\end{align}

\begin{example}
(\textsc{Unit Square}) A prototype domain of the shape of a unit-square centered around the origin, $\mathcal{F}^0\!=\!\left\lbrace (p,q)\left|\,|p|\!\leq\!1,\, |q|\!\leq\! 1\right.\right\rbrace$, can be expressed in the form of \eqref{E:proto} with $A\!=\!\!\begin{bmatrix}
1&-1&0&0\\0&0&1&-1
\end{bmatrix}^T\!$, $b\!=\!\!\begin{bmatrix}
1&1&1&1
\end{bmatrix}^T\!$.\hfill\qed\end{example}

\subsection{Constructing Outer Approximation}

An outer approximation of the compact flexibility domain in \eqref{E:flex_compact}, using a \textit{homothet} of the prototype in \eqref{E:proto} amounts to finding positive scalar $\overline{\alpha}$ and 2-dimensional vector $\overline{\beta}\!=\!\begin{bmatrix}
\overline{\beta_p} &\overline{\beta_q}
\end{bmatrix}^T$ such that the following set inclusion condition is satisfied,
\begin{align}
\mathcal{H}\!\left[\overline{\alpha},\overline{\beta};\mathcal{F}^0\right]&\!=\!\left\lbrace (p,q)\left|\,\begin{bmatrix}
p\\q
\end{bmatrix}\!=\!\overline{\alpha}\!\begin{bmatrix}
p^0\\q^0
\end{bmatrix}\!+\!\!\begin{bmatrix}
\overline{\beta_p}\\\overline{\beta_q}
\end{bmatrix}\!,\,(p^0\!,q^0)\!\in\!\mathcal{F}^0\! \right.\right\rbrace\notag\\
&\!=\!\left\lbrace (p,q)\left|\,A\begin{bmatrix}
p\\q
\end{bmatrix}\!\leq\!\overline{\alpha}\,b\!+\!A\!\begin{bmatrix}
\overline{\beta_p}\\\overline{\beta_q}
\end{bmatrix}\! \right.\right\rbrace\!\supseteq\mathcal{F}.
\end{align}
This condition translates into a set of semi-algebraic conditions as follows:
\begin{align*}
\forall i\!\in\!\lbrace 1,\dots,n\rbrace:~a_{ip}\,(\overline{\beta_p}\!-\!p)+a_{iq}\,(\overline{\beta_q}\!-\!q)\!+\!\overline{\alpha}\,b_i\!\geq\!0\text{ on }\mathcal{F}.
\end{align*}
After relaxing the inequality `$\geq$' to strictly inequality `$>$' (note that for convex domains, for any $\alpha$ satisfying the inequality condition, we can always find an infinitesimally larger $\alpha$ that satisfies the strict inequality), we can use Theorem\,\ref{T:Putinar} to construct the following SOS optimization problem
\begin{align}\label{E:outer}
&\text{minimize}~\overline{\alpha}\!>\!0,\,\\
&\text{s.t., }\notag\\
&\overline{\beta_p}\!\in\!\mathbb{R},\,\overline{\beta_q}\!\in\!\mathbb{R},\notag\\
&\sigma_{ij}^k\!\in\!\Sigma[p,q]\,~\forall i\!\in\!\lbrace 1,\dots,n\rbrace,\,j\!\in\!\lbrace 1,\dots,m^k\rbrace,\,k\!\in\!\lbrace 1,\dots,K\rbrace\notag\\
&\forall k\!:\left\lbrace\!\! \begin{array}{r}\begin{bmatrix}
a_{1p} & a_{1q}
\end{bmatrix}\begin{bmatrix}
\overline{\beta_p}\!-\!p\\
\overline{\beta_q}\!-\!q
\end{bmatrix}\!+\!\overline{\alpha}\,b_1\!-\!\sum_{j=1}^{m^k}\sigma_{1j}\,g_j^k\in\Sigma[p,q]\\
\vdots\qquad\qquad \qquad\qquad \\
\begin{bmatrix}
a_{np} & a_{nq}
\end{bmatrix}\begin{bmatrix}
\overline{\beta_p}\!-\!p\\
\overline{\beta_q}\!-\!q
\end{bmatrix}\!+\!\overline{\alpha}\,b_n\!-\!\sum_{j=1}^{m^k}\sigma_{nj}\,g_j^k\in\Sigma[p,q]\end{array}\!\!\right.\notag
\end{align}
Note that the SOS conditions in the above SOS problem are affine in the decision variables ($\overline{\alpha},\,\overline{\beta_p},\,\overline{\beta_q},\,\sigma_{ij}^k$) and hence can be solved directly. Note that the value of $\overline{\alpha}$ represents how \textit{large} the homothet is, and therefore by solving for the minimal $\overline{\alpha}$ we ensure that the \textit{outer approximation} is the tightest.

\begin{remark}
The problem formulated in \eqref{E:outer} can be used to compute the outer approximation for both the continuous flexibility domains as well as the discrete flexibility domains of the type shown in Fig.\ref{F:hvac} (and described in Example\,\ref{Ex:hvac}).
\end{remark}

\subsection{Constructing Inner Approximation}

Solving for the \textit{inner approximation} is not as straightforward. Given a scaling factor $\underline{\alpha}$ and a translational vector $\underline{\beta}$\,, checking whether the homothet  $\mathcal{H}\!\left[\underline{\alpha},\underline{\beta};\mathcal{F}^0\right]$ is included completely inside the flexibility domain $\mathcal{F}$ can be done by solving the following feasibility problem
\mysubeq{}{
&\text{given}~\underline{\alpha}\!>\!0,\,\underline{\beta_p}\!\in\!\mathbb{R},\,\underline{\beta_q}\!\in\!\mathbb{R},\\
&\text{check:}~\forall (k,i)\,g_i^k\geq 0\text{ on }\!\!\left\lbrace\! (p,q)\!\left|\,A\!\begin{bmatrix}
p\\q
\end{bmatrix}\!\leq\!\overline{\alpha}\,b\!+\!A\!\begin{bmatrix}
\overline{\beta_p}\\\overline{\beta_q}
\end{bmatrix}\! \right.\right\rbrace
}
which can be cast into an SOS feasibility problem (for some given $\underline{\alpha}\!>\!0\,,\,\underline{\beta_p}\!\in\!\mathbb{R}$ and $\underline{\beta_q}\!\in\!\mathbb{R}$):
\begin{align}\label{E:inner}
&\text{find}~\sigma_{ij}^k\!\in\!\Sigma[p,q]\,~\forall \left\lbrace\begin{array}{l}i\!\in\!\lbrace 1,\dots,m^k\rbrace,\,\\
j\!\in\!\lbrace 1,\dots,n\rbrace,\,\\
k\!\in\!\lbrace 1,\dots,K\rbrace\end{array}\right.\\
&\text{s.t. }\notag\\
&\forall k\!:\!\left\lbrace\!\!\begin{array}{r}g_1^k\!-\!\sum_{j=1}^n\sigma_{1j}^k(\begin{bmatrix}
a_{jp}~\, a_{jq}
\end{bmatrix}\begin{bmatrix}
\underline{\beta_p}\!-\!p\\
\underline{\beta_q}\!-\!q
\end{bmatrix}\!+\!\underline{\alpha}\,b_j)\in\Sigma[p,q]\\
\vdots\qquad\qquad \qquad\qquad\qquad\\
g_{m^k}^k\!-\!\sum_{j=1}^n\sigma_{m^kj}^k(\begin{bmatrix}
a_{jp}~\, a_{jq}
\end{bmatrix}\begin{bmatrix}
\underline{\beta_p}\!-\!p\\
\underline{\beta_q}\!-\!q
\end{bmatrix}\!+\!\underline{\alpha}\,b_j)\in\Sigma[p,q]\end{array}\right.\notag
\end{align}
Notice that the SOS constraints in \eqref{E:inner} are bilinear in the decision polynomial variables $\sigma_{ij}^k$ and the \textit{homothet} parameters $\underline{\alpha}$ and $\underline{\beta}$. Therefore this cannot be solved directly to find out the \textit{best} inner approximation, and can only be used to check whether a given homothet is an inner approximation or not. In this paper, we propose a heuristic to find the homothet parameters that define the best inner approximation. Note that given a $\underline{\beta}$ the best inner approximation is the largest \textit{homothet} that is contained inside the flexibility domain, which can be solved by running a bisection search for the maximum $\underline{\alpha}$ for which \eqref{E:inner} is feasible. The value of this largest $\underline{\alpha}$ can be further improved if we can identify a direction in which to \textit{translate} the \textit{homothet}, by detecting the edges of the polytope which touch the flexibility domain boundary (i.e. \textit{binding} edges). The \textit{best} inner approximation is computed in a two-stage iterative process as described below.

\subsubsection{Step 1} 
Find $\underline{\alpha}$ given $\underline{\beta}$\,. Given a $\underline{\beta}$ we can 
run a bisection search for the largest positive scalar $\underline{\alpha}$ that solves \eqref{E:inner}.

\subsubsection{Step 2}
Update $\underline{\beta}$\,. Once we have found the largest homothet $\mathcal{H}\!\left[\underline{\alpha},\underline{\beta};\mathcal{F}^0\right]$ of the \textit{prototype} polygon, that is contained wholly inside the flexibility domain, we can identify the edges of the \textit{homothet} which lie on the boundary of the flexibility domain. The edges represent the \textit{binding} constraints that limit the expansion of the inner \textit{homothet}. Thus we choose a direction $\hat{b}$ that is a vector sum of all the directions orthogonal to the \textit{binding} edges, and choose a new $\underline{b}\leftarrow \underline{b}+\varepsilon\hat{b}$\,, for some sufficiently small scalar $\varepsilon$\,. Checking which edges are binding can be done by solving appropriate sum-of-squares problems, which we have omitted for brevity.

These two steps are repeated until a convergence in the value of $\underline{\alpha}$ is achieved. 
Let us denote by $\mathcal{H}\!\left[\underline{\alpha^*},\underline{\beta^*};\mathcal{F}^0\right]$ the \textit{largest} homothet contained inside the flexibility domain, such that there exists no $\underline{\alpha}>\underline{\alpha^*}$ and $\underline{\beta}$ for which $\mathcal{H}\!\left[\underline{\alpha},\underline{\beta};\mathcal{F}^0\right]$ is an inner approximation. Let us assume that $\mathcal{H}\!\left[\underline{\alpha^0},\underline{\beta^0};\mathcal{F}^0\right]$ is an inner approximation for some translational vector $\underline{\beta^0}$ and scaling factor $\underline{\alpha^0}$. Then the success (convergence) of the above two-stage algorithm depends on the following condition:
\begin{assumption}\label{AS:monotonic}
$\underline{\alpha}(\epsilon)$ increases monotonically with $\epsilon\in[0,1]$\,, where $\underline{\alpha}(\epsilon)$ is the largest scaling factor corresponding to the translational vector $\underline{\beta}(\epsilon)=\underline{\beta^0}+\epsilon\left(\underline{\beta^*}-\underline{\beta^0}\right)$\,.
\end{assumption}

Note that the inner approximations do not exist for DERs that have discrete flexibility domains. Hence the above method is only applicable to continuous flexibility domains.

\section{Quality of Approximation}\label{S:metric}

In order for the approximation of the flexibility domains to be used in network-level resource optimization, it is important that the approximations are \textit{close} to the actual flexibility domains. In this section, we propose metric(s) to quantify the \textit{closeness} of the approximation. Such metrics could be useful to choose from a set of \textit{prototypes} to approximate the flexibility domain of a given DER. Furthermore, such metrics could be passed up to the network optimizer, along with the approximating \textit{homothets}, to facilitate better resource allocation. In this paper, we discuss two such possible metrics. First one is related to the worst case error in a dispatched point (from network operator) and a feasible solution:
\begin{definition}\label{D:distance}
The distance metric $\pi^d(\overline{\mathcal{F}},\mathcal{F})$ for two sets $\overline{\mathcal{F}}\supset\mathcal{F}$ is defined as the maximal distance between any point in $\overline{\mathcal{F}}\backslash\mathcal{F}$ and its nearest feasible point in $\mathcal{F}$\,.
\end{definition}
%
\begin{proposition}\label{P:distance}
The distance metric for the outer and inner approximations of the type $\overline{\mathcal{H}}=\mathcal{H}\!\left[\overline{\alpha},\overline{\beta};\mathcal{F}^0\right]$ and $\underline{\mathcal{H}}=\mathcal{H}\!\left[\underline{\alpha},\underline{\beta};\mathcal{F}^0\right]$\,, where $\mathcal{F}^0$ is a polygon, is given by 
\begin{align*}
\pi^d(\overline{\mathcal{H}},\underline{\mathcal{H}})=\max_i \left\lVert (\overline{\alpha}-\underline{\alpha})\,v_i + \overline{\beta}-\underline{\beta}\right\rVert_2
\end{align*}
where $v_i\,\forall i$ is the vector representing the vertices of $\mathcal{F}^0$.
\end{proposition}
\begin{proof}
Since homothetic transformation include only scaling and translation, the edges in both the inner and outer homothets are parallel to each other. Therefore the maximal distance between the two polygons occur at the respective vertices.\hfill\qed
\end{proof}
The second metric is related to the likelihood of a dispatched point lying inside the actual flexibility domain:
\begin{definition}\label{D:area}
The area metric $\pi^a(\overline{\mathcal{F}},\mathcal{F})$ for the outer approximated and actual flexibility domains, $\overline{\mathcal{F}}\supset\mathcal{F}$\,, is defined as the likelihood that an optimally dispatched point in $\overline{\mathcal{F}}$ also lies inside $\mathcal{F}$\,.
\end{definition}

To compute this metric one needs to assign a probability of each point in the approximated domains being dispatched by the network operator. In absence of any prior knowledge of the probability, one can start by assuming that each point in the approximated domain is \textit{equally likely} to be dispatched. Under such an assumption, we can compute the area metric as follows:
\begin{proposition}\label{P:area}
The area metric for the outer and inner approximations of the type $\overline{\mathcal{H}}=\mathcal{H}\!\left[\overline{\alpha},\overline{\beta};\mathcal{F}^0\right]$ and $\underline{\mathcal{H}}=\mathcal{H}\!\left[\underline{\alpha},\underline{\beta};\mathcal{F}^0\right]$\,, is given by $\pi^a(\overline{\mathcal{H}},\underline{\mathcal{H}})=\left({\underline{\alpha}}/{\overline{\alpha}}\right)^2$\,.
\end{proposition}
\begin{proof}
Note that the area under a homothet $\mathcal{H}\!\left[{\alpha},{\beta};\mathcal{F}^0\right]$ is equal to $\alpha^2$ times the area under the prototype $\mathcal{F}^0$\,. The rest follows directly.\hfill\qed
\end{proof}
Note that the larger the value of the area metric, the better are the domain approximations, while for the distance metric the smaller value implies better approximation. These metrics are applicable directly to continuous flexibility domains. For discrete domains, alternative methods need to be devised to compute the metrics (not covered in this article).

\section{Numerical Results}\label{S:result}

\begin{figure*}[thpb]
\centering
\captionsetup{justification=centering}
\subfigure[battery (DER 1)]{
\includegraphics[scale=0.27]{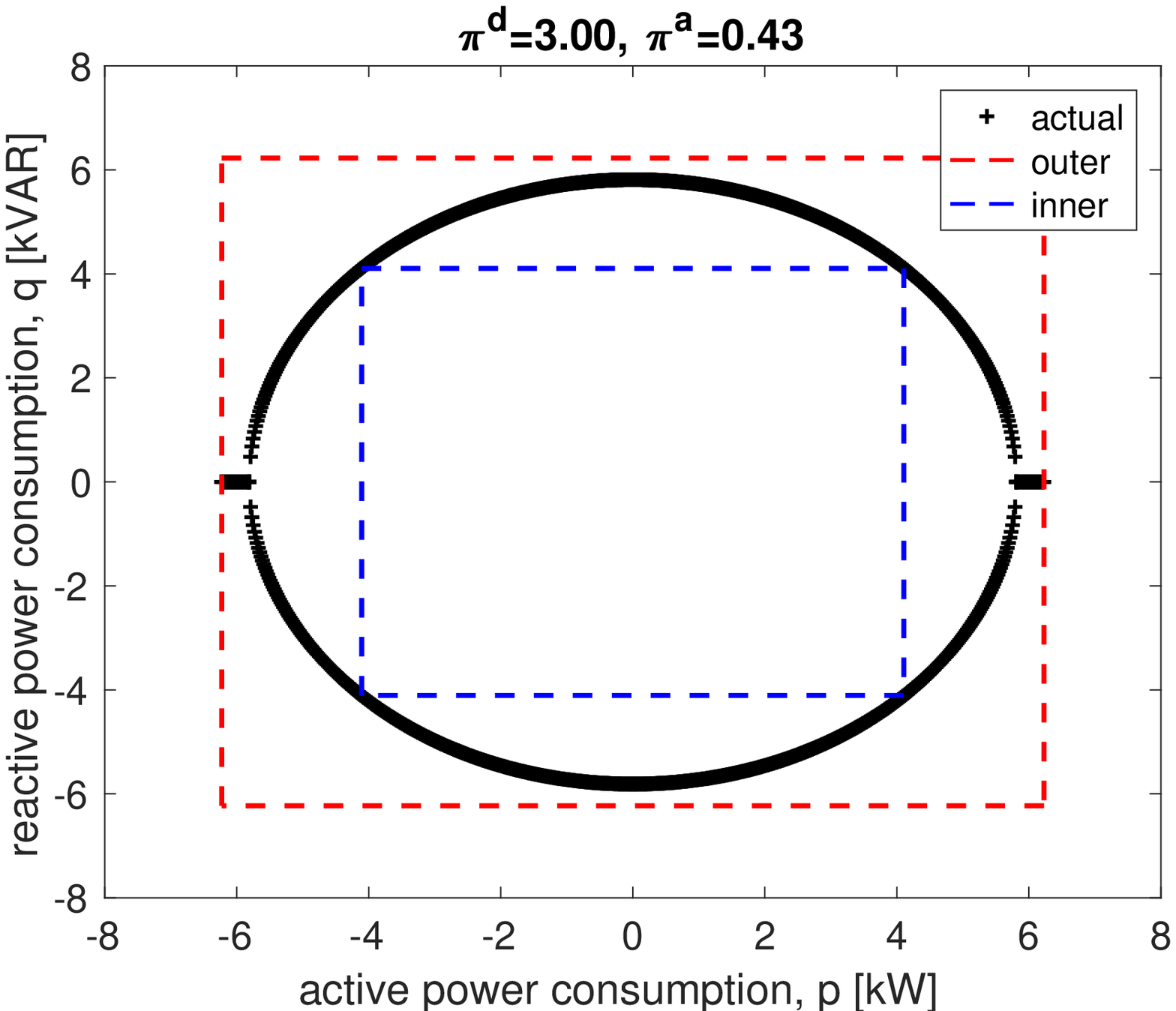}\label{F:box_1}
}
\hspace{0.1in}
\subfigure[solar PV inverter (DER 2)]{
\includegraphics[scale=0.27]{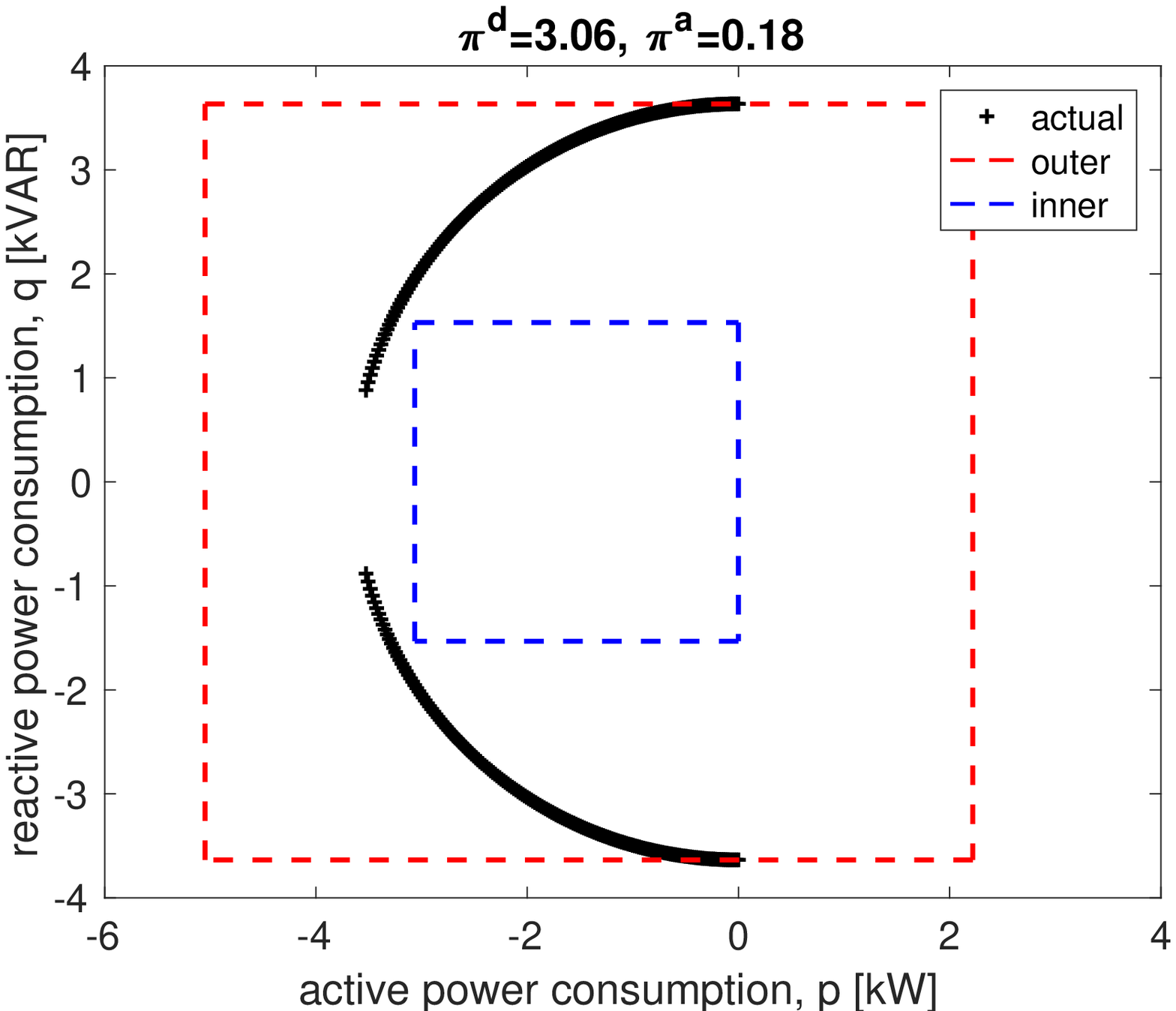}\label{F:box_2}
}
\hspace{0.1in}
\subfigure[wind inverter (DER 3)]{
\includegraphics[scale=0.27]{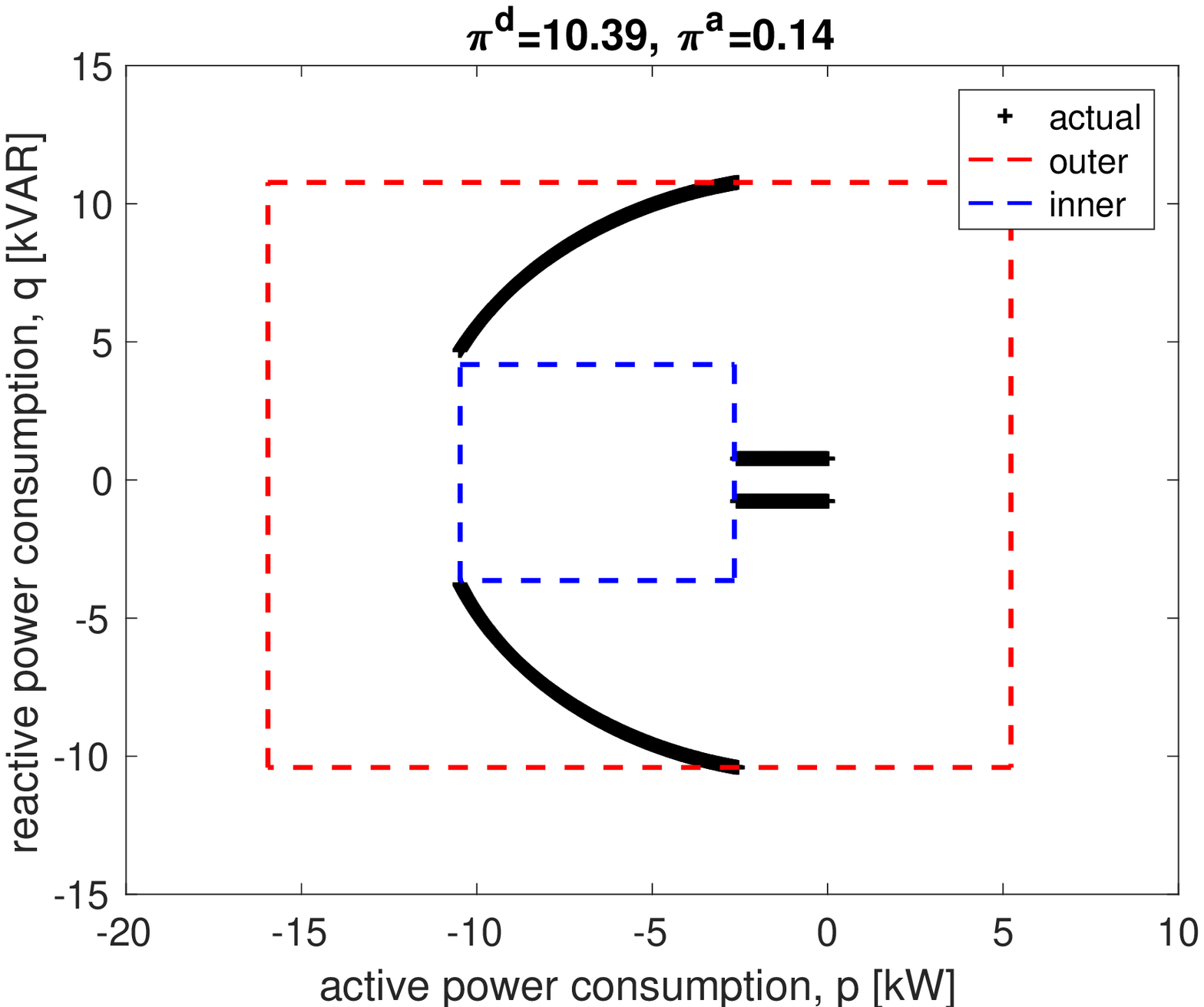}\label{F:box_3}
}
\hspace{0.1in}
\subfigure[battery (DER 4)]{
\includegraphics[scale=0.27]{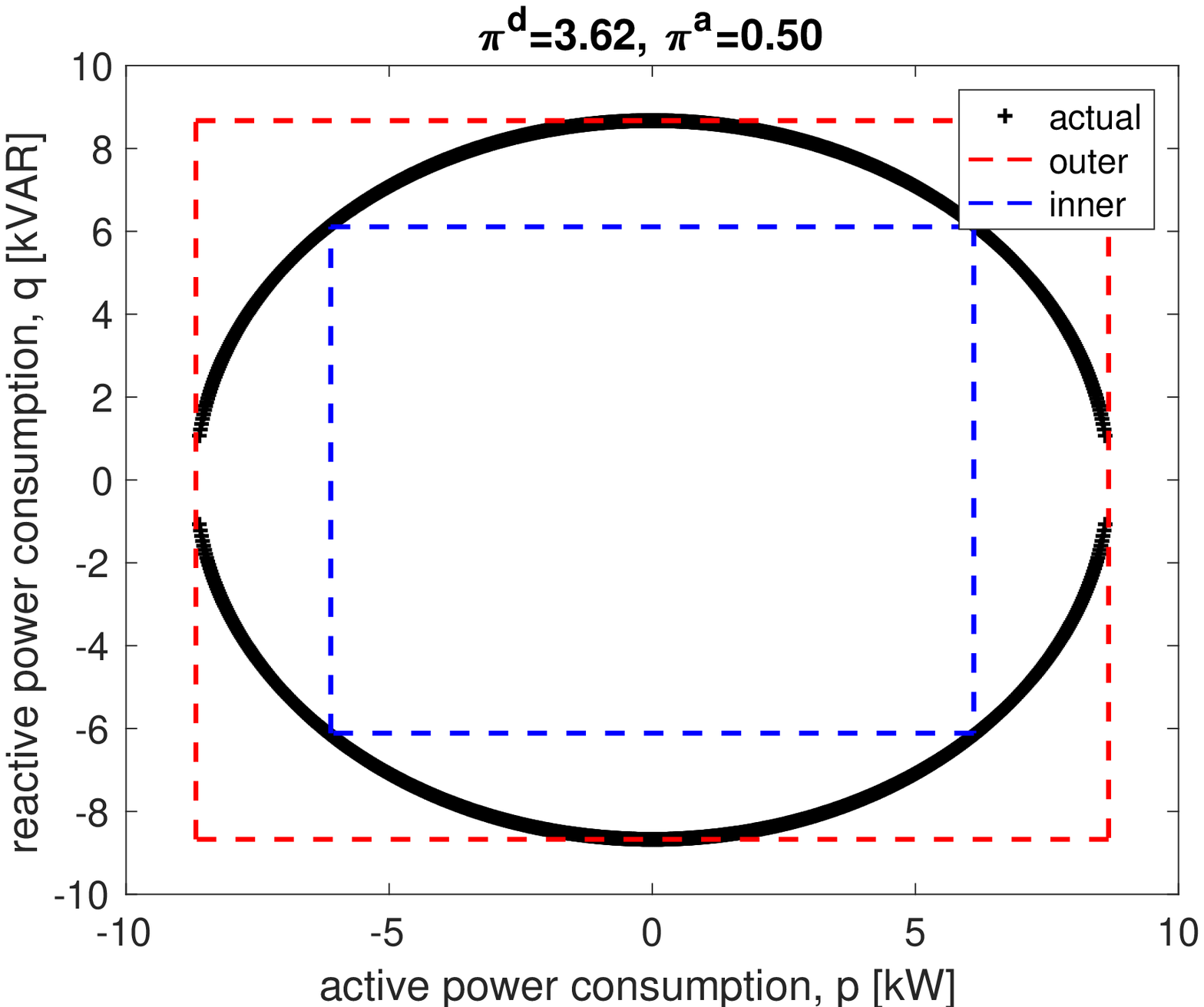}\label{F:box_4}
}
\hspace{0.1in}
\subfigure[wind inverter (DER 5)]{
\includegraphics[scale=0.27]{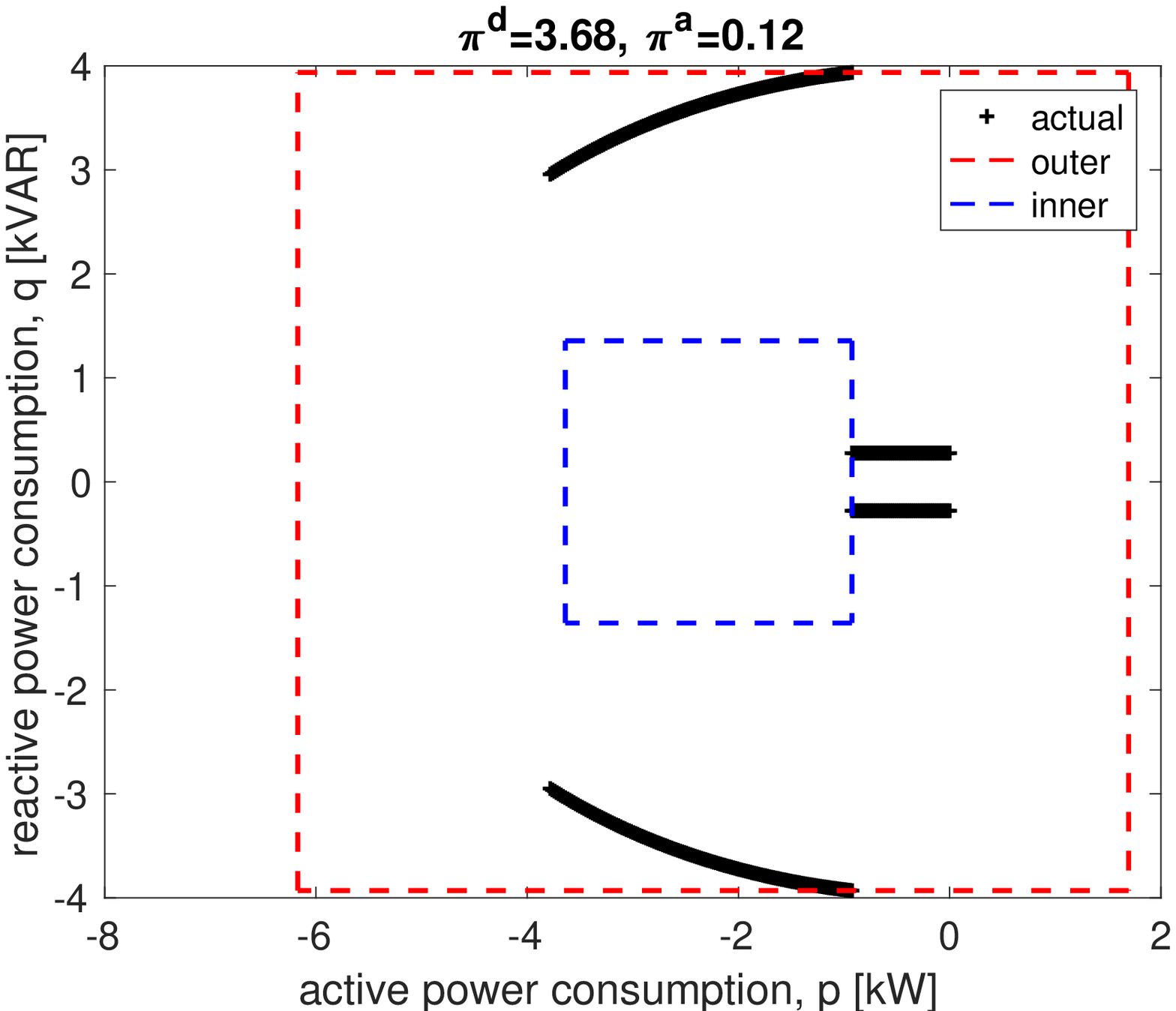}\label{F:box_5}
}
\hspace{0.1in}
\subfigure[aggregated flexibility]{
\includegraphics[scale=0.27]{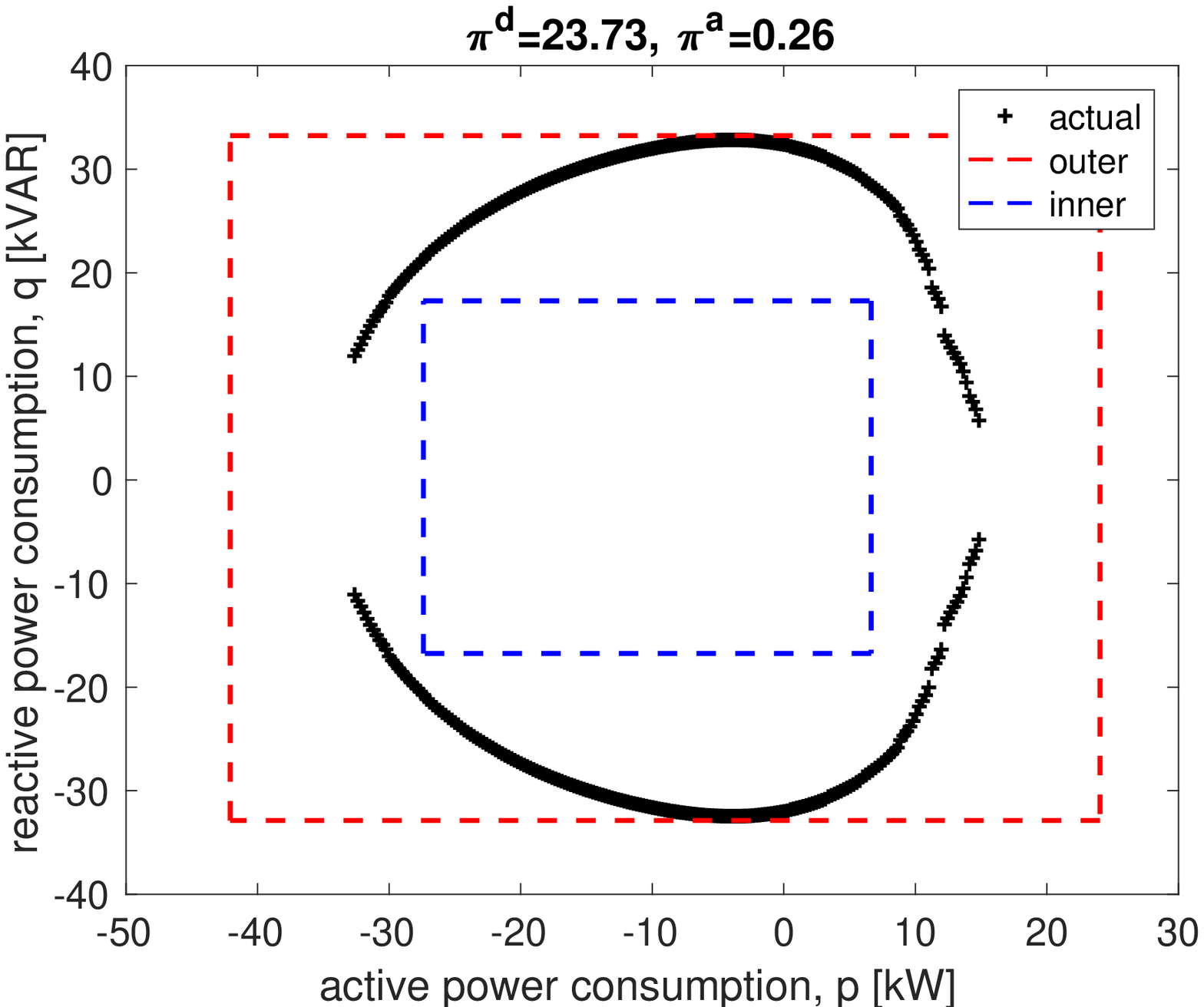}\label{F:box_aggr}
}
\caption[Optional caption for list of figures]{Outer and inner approximations of a collection of PV inverters, wind inverters and batteries, using unit square as a prototype.}
\label{F:box}
\end{figure*}
In this section, we present some numerical results to demonstrate the aggregation of flexibility domains of different types of DERs. Fig.\,\ref{F:box} shows the plots of outer and inner approximations of a collection of five DERs (consisting of PV inverters, wind inverters and batteries), using unit square as a prototype, including the approximation of their Minkowski sum. The values of the distance metric and the area metric for all these approximations are also computed and shown on the plots. 
\begin{figure*}[thpb]
\centering
\captionsetup{justification=centering}
\subfigure[solar PV inverter (DER 2)]{
\includegraphics[scale=0.27]{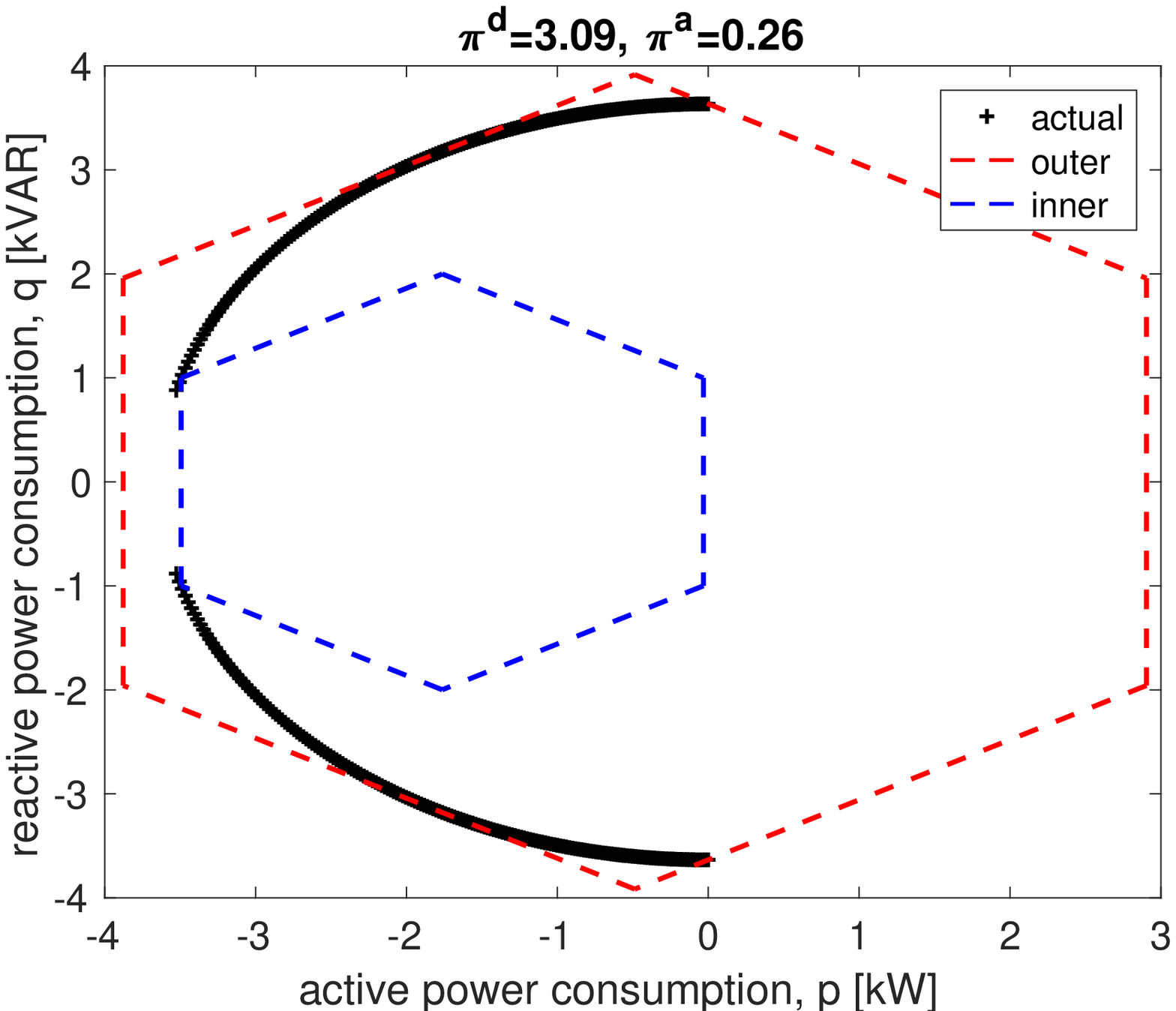}\label{F:hex_2}
}
\hspace{0.1in}
\subfigure[wind inverter (DER 5)]{
\includegraphics[scale=0.27]{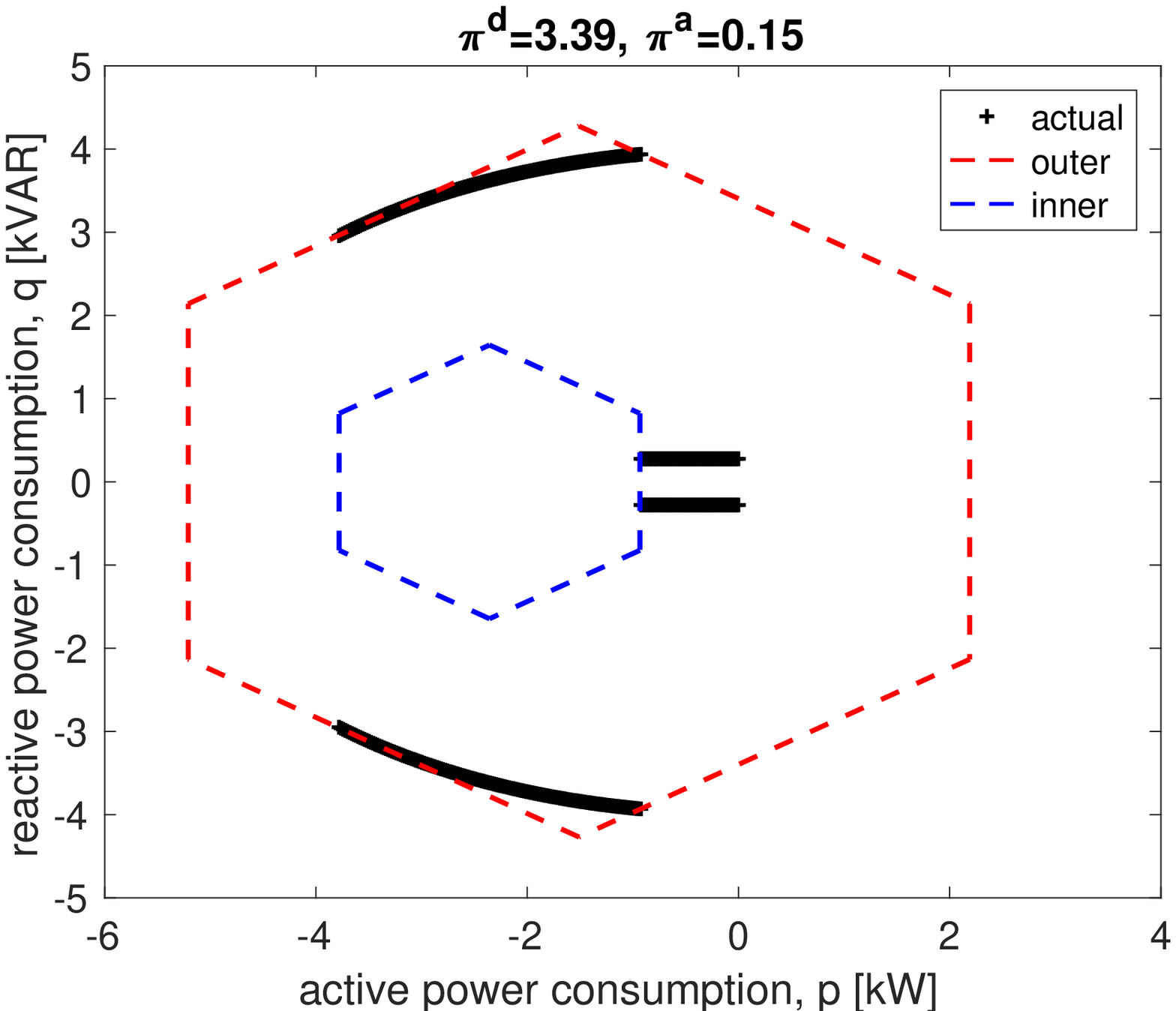}\label{F:hex_5}
}
\hspace{0.1in}
\subfigure[aggregated flexibility]{
\includegraphics[scale=0.27]{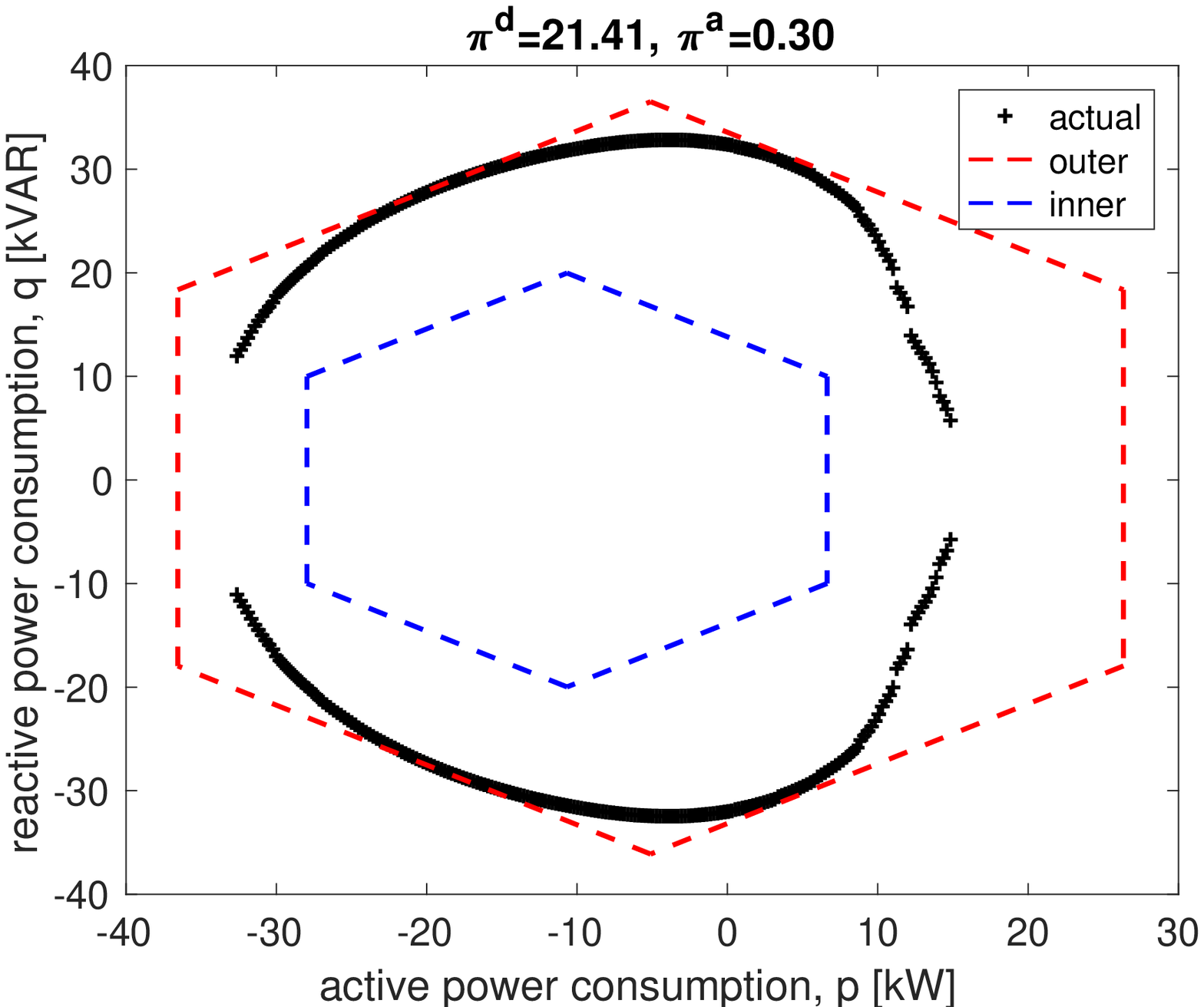}\label{F:hex_aggr}
}
\caption[Optional caption for list of figures]{Outer and inner approximations of the same set of DERs, using regular unit hexagon as a prototype.}
\label{F:hex}
\end{figure*}
Fig.\,\ref{F:hex} shows the approximations using a regular unit hexagon as the prototype, for the same set of devices (only a couple of them shown on the figure). In this particular example, the approximations using a regular hexagon as the prototype turns out to be better than the approximations using unit square as the prototype, according to both the area metric and the distance metric. For example, the value of the area metric of the approximation of the aggregated flexibility turns out to be higher for the regular hexagon prototype ($\pi^a=0.30$) than for the square prototype ($\pi^a=0.26$), while the value of the distance metric is lower with regular hexagon as the prototype ($\pi^d \approx 21.4$) than with a square prototype ($\pi^d \approx 23.7$).

\begin{figure*}[thpb]
\centering
\captionsetup{justification=centering}
\subfigure[individual air-conditioner]{
\includegraphics[scale=0.27]{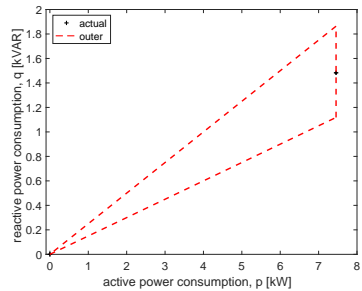}\label{F:tri_1}
}
\hspace{0.1in}
\subfigure[collection of air-conditioners]{
\includegraphics[scale=0.27]{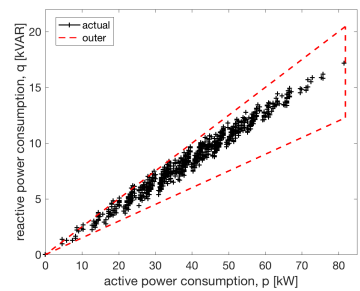}\label{F:tri_aggr}
}
\caption[Optional caption for list of figures]{Outer approximations of a collection of ten residential air-conditioners with discrete flexibility domains, using a triangle as a prototype.}
\label{F:tri}
\end{figure*}

Finally Fig.\,\ref{F:tri} shows the results from approximating the aggregated flexibility of a collection of discrete loads (residential air-conditioners) via a triangle as a prototype. While the approximation of an individual DER results in a very sparse approximated flexibility domain, the approximation of the aggregated flexibility domain tends to be much denser. While a specific quantification of the accuracy (or, quality) of the approximation for discrete DERs is not provided in this article, Fig.\,\ref{F:tri} shows qualitatively how the larger population size of an ensemble of switching loads can lead to better approximation of the aggregated flexibility.

\section{Conclusions}\label{S:concl}
We envision a distribution system optimization framework in which thousands of these DERs are scheduled to operate the system efficiently. Heterogeneity of the DERs, manifested in the different forms of their flexibility domains, pose a serious challenge to the scalability of such an optimization problem. Furthermore, any such optimization framework needs to be able to accommodate a plug-and-play approach whereby any new DER can be easily integrated into the system. In this work, we presented a geometric approach to \textit{homogenizing} the flexibility domains of different types of DERs via \textit{homothets} of some \textit{prototype} domain. We further looked into the choice of the prototype domain, and the choice of the parameters of the homothets, such that the outer (and inner, when applicable) approximations of the flexibility domains are \textit{good}. Simulation results are provided to illustrate the approach. Future efforts will concentrate on extending this framework to model the dynamic evolution of the flexibility, as well as modeling the uncertainties in the flexibility of DERs.

\section*{Acknowledgment}

The authors would like to thank Dr. Andrey Bernstein of National Renewable Energy Laboratory, USA, for his helpful remarks and discussions on this work.

\IEEEtriggeratref{8}



\bibliographystyle{IEEEtran}
\bibliography{IEEEabrv,RefList,references,MyReferences}


%

\end{document}